\documentclass[9pt,twocolumn]{article}
\usepackage[fleqn]{amsmath}
\usepackage{amsthm}
\usepackage[varg]{txfonts}
\usepackage[dvips]{graphicx}
\usepackage{subfigure}
\usepackage{psfrag}

\newtheorem{theorem}{Theorem}
\newtheorem{corollary}[theorem]{Corollary}
\newtheorem{proposition}[theorem]{Proposition}

\title{Graph isomorphism completeness for trapezoid graphs}
\author{Asahi Takaoka\footnote{
The author is with the 
Department of Communications and Computer Engineering, 
Tokyo Institute of Technology, Tokyo 152-8550, Japan. 
E-mail: asahi@eda.ce.titech.ac.jp.}}

\begin{document}
\maketitle
\begin{abstract}
The complexity of the graph isomorphism problem for trapezoid graphs 
has been open over a decade. 
This paper shows that the problem is GI-complete. 
More precisely, we show that the graph isomorphism problem is GI-complete 
for comparability graphs of partially ordered sets 
with interval dimension 2 and height 3. 
In contrast, the problem is known to be solvable in polynomial time 
for comparability graphs of partially ordered sets 
with interval dimension at most 2 and height at most 2. 

\end{abstract}
\section{Introduction}
Let $G$ be an undirected simple graph, and 
let $V(G)$ and $E(G)$ be the vertex set and the edge set of $G$, respectively. 
Two graphs $G_1$ and $G_2$ are said to be \emph{isomorphic} 
if there is a bijection $\phi:V(G_1) \to V(G_2)$ such that 
for every pair of vertices $u, v \in V(G_1)$, 
$uv \in E(G_1)$ if and only if $\phi(u)\phi(v) \in E(G_2)$. 
Such $\phi$ is called \emph{isomorphism} from $G_1$ to $G_2$. 
We denote by $G_1 \simeq G_2$ if $G_1$ and $G_2$ are isomorphic. 
The \emph{graph isomorphism problem} asks 
whether two given graphs are isomorphic. 
Although the problem is in NP, 
it is not known to be NP-complete or polynomial-time solvable. 

The graph isomorphism problem 
for particular classes of graphs has been investigated. 
See~\cite{BC79-TR,Spinrad03,Uehara14-DMTCS} for survey. 
The problem for a graph class is said to be \emph{GI-complete} 
if it is polynomial-time equivalent to the problem for general graphs. 
For some graph classes, the problem 
is not known to be GI-complete or polynomial-time solvable. 
One of such graph classes is 
trapezoid graphs~\cite{Spinrad03,Uehara14-DMTCS,UTN05-DAM}. 
Trapezoid graphs are natural generalization of 
interval graphs and permutation graphs, for which the problem can be solved 
in linear time~\cite{Colbourn81-Networks,LB79-JACM,Spinrad85-SIAMCOMP}. 
We show in this paper that the problem is GI-complete for trapezoid graphs. 

\section{GI-completeness for trapezoid graphs}
Let $L_1$ and $L_2$ be two lines parallel to $x$-axis in the $xy$-plane. 
A graph $G$ is called a \emph{trapezoid graph}~\cite{CK87-CN,DGP88-DAM} 
if for each vertex $v \in V(G)$, there is a trapezoid $T_v$ 
with parallel sides along $L_1$ and $L_2$ such that 
for any pair of vertices $u, v \in V(G)$, $uv \in E(G)$ if and only if 
$T_u$ and $T_v$ intersect. 
The set $\{T_v \mid v \in V(G)\}$ is called 
a \emph{trapezoid representation} of $G$. 

The \emph{complement} of a graph $G$ is the graph $\bar{G}$ 
such that $V(\bar{G}) = V(G)$ and for any pair of vertices $u, v \in V(G)$, 
$uv \in E(\bar{G})$ if and only if $uv \notin E(G)$. 
Notice that for any two graphs $G_1$ and $G_2$, 
$G_1 \simeq G_2$ if and only if $\bar{G_1} \simeq \bar{G_2}$. 
To prove the GI-completeness of trapezoid graphs, 
we consider the complements of trapezoid graphs, 
which are known as comparability graphs of partially ordered sets 
with interval dimension at most 2~\cite{DGP88-DAM}. 

A \emph{partially ordered set} (\emph{poset} for short) is 
a pair $P = (X, \preceq)$, 
where $X$ is a finite set and $\preceq$ is a binary relation on $X$ 
that is reflexive, antisymmetric, and transitive. 
We denote $x \prec y$ if $x \preceq y$ and $x \neq y$. 
Two elements $x, y \in X$ are said to be \emph{comparable} in $P$ 
if either $x \prec y$ or $x \succ y$, and 
are said to be \emph{incomparable} otherwise. 
A subset $Y \subseteq X$ is called a \emph{chain} of $P$ 
if any pair of elements of $Y$ are comparable in $P$. 
A chain of $P$ is \emph{maximum} 
if no other chain contains more elements than it, and 
the \emph{height} of $P$ is the number of elements in a maximum chain of $P$. 

A poset $P = (X, \preceq)$ is called an \emph{interval order} 
if for each element $x \in X$, there is an interval $I_x = [l_x, r_x]$ 
on the real line such that for any pair of elements $x, y \in X$, 
$x \prec y$ if and only if $r_x < l_y$. 
Here, we use $<$ to denote the ordering of points on the real line, 
while $\prec$ indicates the relation of a poset. 
The set $\{I_x \mid x \in X\}$ is called 
an \emph{interval representation} of $P$. 
A family $\{ P_i = (X, \preceq_i) \mid 1 \leq i \leq k \}$ of posets 
on the same set is said to \emph{realize} a poset $P = (X, \preceq)$ 
if for any $x, y \in X$, 
$x \preceq y$ if and only if $x \preceq_i y$ 
for every $i \in \{1, 2, \ldots, k\}$. 
The \emph{interval dimension} of a poset $P$ 
is the minimum number $k$ of interval orders that realize $P$. 

A graph $G$ is called a \emph{comparability graph} 
of a poset $P = (X, \preceq)$ 
if there is a bijection assigning each vertex $v \in V(G)$ 
to an element $x_v \in X$ 
such that for any $u, v \in V(G)$, $uv \in E(G)$ 
if and only if $x_u$ and $x_v$ are comparable in $P$. 
We have the following, which is proved in the next section. 

\begin{theorem}\label{theorem:GI-trapezoid}
The graph isomorphism problem is GI-complete 
for comparability graphs of posets 
with interval dimension 2 and height 3. 
\qed
\end{theorem}

\begin{figure*}[t]
  \centering
  \psfrag{a1}{$a_1$}
  \psfrag{a2}{$a_2$}
  \psfrag{a3}{$a_3$}
  \psfrag{b1}{$b_1$}
  \psfrag{b2}{$b_2$}
  \psfrag{b3}{$b_3$}
  \psfrag{e1}{$e_1$}
  \psfrag{e2}{$e_2$}
  \psfrag{e3}{$e_3$}
  \psfrag{e4}{$e_4$}
  \psfrag{e5}{$e_5$}
  \psfrag{c1}{$c_1$}
  \psfrag{c2}{$c_2$}
  \psfrag{c3}{$c_3$}
  \psfrag{c4}{$c_4$}
  \psfrag{c5}{$c_5$}
  \psfrag{L1}{$L_1$}
  \psfrag{L2}{$L_2$}
  \psfrag{T1}{\large $T_{c_1}$}
  \psfrag{T4}{\large $T_{c_4}$}
  \psfrag{1a1}{$p_1(a_1)$}
  \psfrag{1a2}{$p_1(a_2)$}
  \psfrag{1a3}{$p_1(a_3)$}
  \psfrag{1b1}{$p_1(b_1)$}
  \psfrag{1b2}{$p_1(b_2)$}
  \psfrag{1b3}{$p_1(b_3)$}
  \psfrag{2a1}{$p_2(a_1)$}
  \psfrag{2a2}{$p_2(a_2)$}
  \psfrag{2a3}{$p_2(a_3)$}
  \psfrag{2b1}{$p_2(b_1)$}
  \psfrag{2b2}{$p_2(b_2)$}
  \psfrag{2b3}{$p_2(b_3)$}
  \psfrag{1l1}{$l_1(c_1)$}
  \psfrag{1r1}{$r_1(c_1)$}
  \psfrag{1I1}{$I_1(c_1)$}
  \psfrag{1l4}{$l_1(c_4)$}
  \psfrag{1r4}{$r_1(c_4)$}
  \psfrag{1I4}{$I_1(c_4)$}
  \psfrag{2l1}{$l_2(c_1)$}
  \psfrag{2r1}{$r_2(c_1)$}
  \psfrag{2I1}{$I_2(c_1)$}
  \psfrag{2l4}{$l_2(c_4)$}
  \psfrag{2r4}{$r_2(c_4)$}
  \psfrag{2I4}{$I_2(c_4)$}
  \psfrag{o}{$o$}
  \begin{minipage}{0.5\hsize}
    \centering
    \subfigure[A bipartite graph $G_0$ with bipartition $(\{a_1, a_2, a_3\}, \{b_1, b_2, b_3\})$.]{\includegraphics[scale=.6]{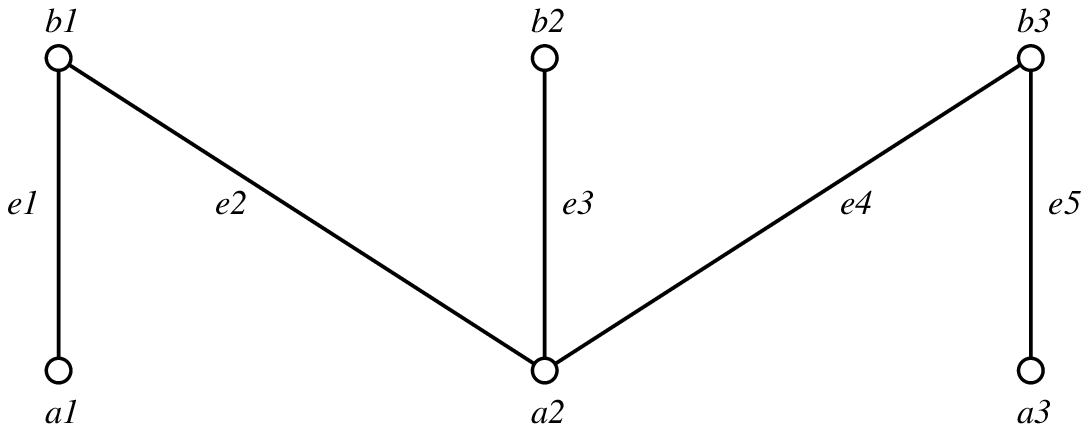}\label{fig:G0}}
    \subfigure[The graph $H_0$ obtained from $G_0$. 
      The vertex $c_i$, $1 \leq i \leq 5$, denoted by a gray point, 
      corresponds to the edge $e_i$ of $G_0$.]{\includegraphics[scale=.6]{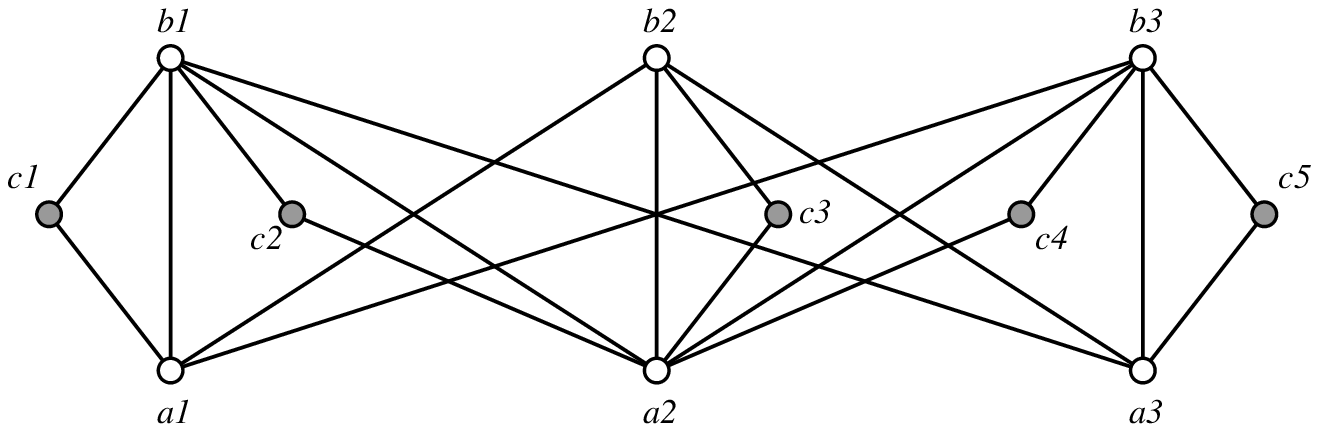}\label{fig:H0}}
  \end{minipage}\hfill
  \begin{minipage}{0.5\hsize}
    \centering
    \subfigure[The interval representations of $P_1$ and $P_2$ that realize $P_{H_0}$ together with the trapezoid representation of the complement of $H_0$. The intervals and trapezoids corresponding to the vertices of $A \cup B$ are degenerated to the points and the line segments, respectively. The intervals and trapezoids corresponding to $c_2, c_3, c_5$ are omitted for the simplicity.]{\includegraphics[scale=.6]{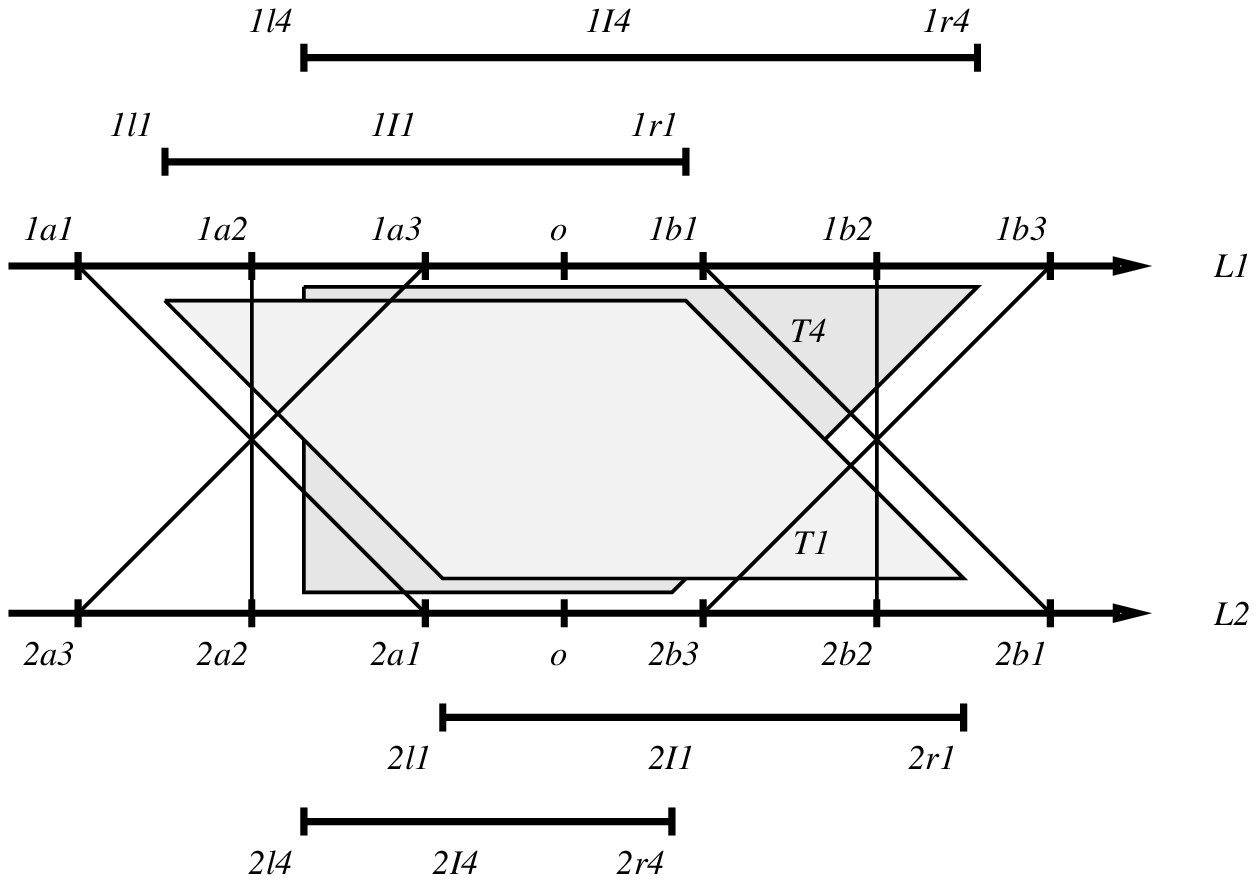}\label{fig:representation}}
  \end{minipage}
  \caption{An example of the construction of the comparability graph (Figs.~\ref{fig:G0} and~\ref{fig:H0}) and the representation of the pair of interval orders that realizes the poset (Fig.~\ref{fig:representation}).}
  \label{fig:figs}
\end{figure*}

Since a graph is a comparability graph of 
a poset with interval dimension at most 2 
if and only if it is the complement of a trapezoid graph~\cite{DGP88-DAM}, 
we have the following. 

\begin{corollary}
The graph isomorphism problem is GI-complete 
for trapezoid graphs. \qed
\end{corollary}

Theorem~\ref{theorem:GI-trapezoid} also gives a dichotomy 
for the graph isomorphism problem for comparability graphs of posets 
with interval dimension at most 2, 
since the problem can be solved in polynomial time 
if the height of the poset is at most~2. 
\begin{proposition}
The graph isomorphism problem can be solved in $O(n^2)$ time 
for comparability graphs of posets 
with interval dimension at most 2 and height at most 2, 
where $n$ is the number of vertices of a graph. 
\end{proposition}
\begin{proof}
The complements of comparability graphs of posets 
with interval dimension at most 2 and height at most 2 are circular-arc graphs 
with clique-cover number~2~\cite{TTU14-IEICE,TM76-DM}, 
for which the graph isomorphism problem can be solved 
in linear time~\cite{CLMNSSS13-DMTCS,Eschen97-phd}. 
Since it requires $O(n^2)$ time to take the complements of graphs, 
we have the proposition. 
\end{proof}

Comparability graphs of posets with interval dimension at most 2 and 
height at most 2 are also known 
as \emph{2-directional orthogonal ray graphs}~\cite{STU10-DAM,TTU14-IEICE}. 
See~\cite{Uehara14-DMTCS} for more information on 
the isomorphism of these graphs. 

\section{Proof of Theorem~\ref{theorem:GI-trapezoid}}
The graph isomorphism problem is GI-complete 
for connected bipartite graphs~\cite{BC79-TR}. 
We show a polynomial-time reduction from the problem 
for connected bipartite graphs to the problem 
for comparability graphs of posets 
with interval dimension 2 and height 3. 
The reduction is similar to that of~\cite{Uehara13-Algorithms,UTN05-DAM}. 

Let $G$ be a connected bipartite graph with bipartition $(A, B)$ 
with $|A|, |B| \geq 3$. 
We construct a graph $H$ from $G$ 
\begin{itemize}
\item by replacing each edge $e = ab$ of $G$ with a vertex $c_e$ 
together with edges $ac_e$ and $bc_e$, and 
\item by adding edges so that the subgraph induced by $A \cup B$ 
is a complete bipartite graph with bipartition $(A, B)$. 
\end{itemize}
See Figs.~\ref{fig:G0} and~\ref{fig:H0} for example of the construction. 
The graph $H$ can be constructed in polynomial time. 
Let $C = \{c_e \mid e \in E(G)\}$, and 
we call $(A, B, C)$ the \emph{tripartition} of $H$. 

We first show that two connected bipartite graphs $G_1$ and $G_2$ 
are isomorphic if and only if $H_1$ and $H_2$ are isomorphic. 
Since it is obvious that $H_1 \simeq H_2$ if $G_1 \simeq G_2$, 
we show the other direction. 
Let $(A_i, B_i, C_i)$ be the tripartition of $H_i$ 
for each $i \in \{1, 2\}$. 
The degree of all vertices of $C_i$ are 2 
and the degree of the other vertices are at least 3, 
since $|A_i|, |B_i| \geq 3$. 
Hence, an isomorphism from $H_1$ to $H_2$ maps 
the vertices of $C_1$ to the vertices of $C_2$ and 
maps the vertices of $A_1 \cup B_1$ to the vertices of $A_2 \cup B_2$. 
Since $G_i$ can be obtained from $H_i$ by 
deleting all edges between the vertices of $A_i \cup B_i$, and 
by deleting each $c \in C_i$ and adding an edge joining two vertices adjacent to $c$, 
we conclude that $G_1 \simeq G_2$ if $H_1 \simeq H_2$. 

We next show that $H$ is the comparability graph of a poset 
with interval dimension at most 2 and height 3. 
Let $P_H= (V(H), \preceq)$ be the poset obtained from $H$ 
with tripartition $(A, B, C)$ such that 
$a \prec b$, $a \prec c$, and $c \prec b$ 
for any $a \in A$, $b \in B$, and $c \in C$. 
It is easy to verify that the relation $\prec$ is transitive 
and the height of the poset $P_H$ is 3. 

Now, it suffices to show the interval representations of 
two interval orders $P_1$ and $P_2$ that realize $P_H$. 
Let $\{I_i(v) \mid v \in V(H)\}$ be the interval representation of $P_i$ 
for each $i \in \{1, 2\}$. 
The interval $I_i(v)$ for any $v \in A \cup B$ 
is degenerated to the point $p_i(v)$ in the representations. 
For any $c \in C_i$, let $l_i(c)$ and $r_i(c)$ be 
the left and right end-point of $I_i(c)$, respectively. 
We denote the representations by a series of points on the real line. 
Let $o$ be the origin of the real line, and let 
$A = \{a_1, a_2, \ldots, a_s\}$ and $B = \{b_1, b_2, \ldots, b_t\}$. 
We place the points corresponding to the elements of $A \cup B$ 
on the real line such that 
%
%
\begin{flalign*}
p_1(a_1) & < p_1(a_2) < \ldots < p_1(a_s) < o \\
& < p_1(b_1) < p_1(b_2) < \ldots < p_1(b_t) 
\end{flalign*}
and 
\begin{flalign*}
p_2(a_s) & < p_2(a_{s-1}) < \ldots < p_2(a_1) < o \\
& < p_2(b_t) < p_2(b_{t-1}) < \ldots < p_2(b_1). 
\end{flalign*}
We can verify that $a \prec b$ for any $a \in A$ and $b \in B$, 
any pair of elements of $A$ are incomparable, and 
any pair of elements of $B$ are incomparable. 

Let $c \in C$ be a vertex of $H$ adjacent to $a_i \in A$ and $b_j \in B$. 
We place the end-points $l_1(c)$ and $r_1(c)$ such that 
\[p_1(a_i) \! < \! l_1(c) \! < \! p_1(a_{i+1}) \leq o \leq p_1(b_{j-1}) \! < \! r_1(c) \! < \! p_1(b_j)\] 
and place the end-points $l_2(c)$ and $r_2(c)$ such that 
\[p_2(a_i) \! < \! l_2(c) \! < \! p_2(a_{i-1}) \leq o \leq p_2(b_{j+1}) \! < \! r_2(c) \! < \! p_2(b_j),\]
where $p_1(a_{s+1}) = p_1(b_0) = o$ and $p_2(a_0) = p_2(b_{t+1}) = o$. 
When more than one vertex is adjacent 
to a vertex $a_i \in A$ (resp. $b_j \in B$), 
we place the left (resp. right) end-points any order 
in the intervals $[p_1(a_i), p_1(a_{i+1})]$ and $[p_2(a_i), p_2(a_{i-1})]$ 
(resp. 
in the intervals $[p_1(b_{j-1}), p_1(b_j)]$ and $[p_2(b_{j+1}), p_2(b_j)]$). 
It can be verified that $a_i \prec c \prec b_j$ and 
$c$ is incomparable to any other element of $A \cup B$. 
Moreover, any pair of elements of $C$ are incomparable, 
since any interval corresponding to an element of $C$ 
contains the origin $o$. 
Hence, the interval orders $P_1$ and $P_2$ realize $P_H$, 
and we have Theorem~1. 

Fig.~\ref{fig:representation} shows 
the interval representations of the pair of posets that realizes $P_{H_0}$. 
The trapezoid representation of the complement of $H_0$ 
is also shown in Fig.~\ref{fig:representation}. 

\section{Concluding remarks}
We show in this paper that 
the graph isomorphism problem is GI-complete for trapezoid graphs. 
Since the problem can be solved in linear time for 
interval graphs~\cite{LB79-JACM} and 
permutation graphs~\cite{Colbourn81-Networks,Spinrad85-SIAMCOMP}, 
it is an interesting open question to determine the complexity of 
the problem for graph classes 
between trapezoid graphs and interval graphs or 
between trapezoid graphs and permutation graphs. 
Examples of such graphs are 
\emph{parallelogram graphs}~\cite{BFIL95-DAM,GMT84-DAM,MSZ11-SIAMCOMP}, 
\emph{triangle graphs}~\cite{CK87-CN,Mertzios12-TCS} and 
\emph{simple-triangle graphs}~\cite{CK87-CN,Mertzios13-LNCS}. 
Other open problems can be found in~\cite{Spinrad03,Uehara14-DMTCS}.


\end{document}